\documentclass[a4paper,UKenglish]{lipics-v2019}

\nolinenumbers
\hideLIPIcs

\usepackage[utf8]{inputenc}

\usepackage{newunicodechar}
\newunicodechar{α}{\ensuremath{\alpha}\xspace}
\newunicodechar{β}{\ensuremath{\beta}\xspace}
\newunicodechar{γ}{\ensuremath{\gamma}\xspace}
\newunicodechar{σ}{\ensuremath{\sigma}\xspace}
\newunicodechar{δ}{\ensuremath{\delta}\xspace}
\newunicodechar{τ}{\ensuremath{\tau}\xspace}
\newunicodechar{φ}{\ensuremath{\varphi}\xspace}
\newunicodechar{ε}{\ensuremath{\varepsilon}\xspace}
\newunicodechar{π}{\ensuremath{\pi}\xspace}
\newunicodechar{ψ}{\ensuremath{\psi}\xspace}
\newunicodechar{θ}{\ensuremath{\theta}\xspace}
\newunicodechar{λ}{\ensuremath{\lambda}\xspace}
\newunicodechar{ξ}{\ensuremath{\xi}\xspace}
\newunicodechar{χ}{\ensuremath{\chi}\xspace}

\newunicodechar{Λ}{\ensuremath{\Lambda}\xspace}
\newunicodechar{Ψ}{\ensuremath{\Psi}\xspace}
\newunicodechar{Γ}{\ensuremath{\Gamma}\xspace}
\newunicodechar{Φ}{\ensuremath{\Phi}\xspace}
\newunicodechar{Π}{\ensuremath{\Pi}\xspace}
\newunicodechar{Δ}{\ensuremath{\Delta}\xspace}
\newunicodechar{Θ}{\ensuremath{\Theta}\xspace}

\newunicodechar{∃}{\ensuremath{\exists}\xspace}
\newunicodechar{∀}{\ensuremath{\forall}\xspace}
\newunicodechar{ℕ}{\ensuremath{\mathbb{N}}\xspace}
\newunicodechar{ℤ}{\ensuremath{\mathbb{Z}}\xspace}
\newunicodechar{⇐}{\ensuremath{\Leftarrow}\xspace}
\newunicodechar{⇔}{\ensuremath{\Leftrightarrow}\xspace}
\newunicodechar{⇒}{\ensuremath{\Rightarrow}\xspace}
\newunicodechar{↦}{\ensuremath{\mapsto}\xspace}

\newunicodechar{∨}{\ensuremath{\lor}\xspace}
\newunicodechar{∧}{\ensuremath{\land}\xspace}
\newunicodechar{⊥}{\ensuremath{\bot}\xspace}

\usepackage{color}

\usepackage{xspace}
\usepackage{amsmath}
\usepackage{amssymb}
\usepackage{multicol}
\usepackage{xcolor}
\usepackage{braket}
\usepackage{mathtools}
\usepackage{enumerate}
\usepackage{tikz}
\usetikzlibrary{chains}
\usetikzlibrary{calc}
\usepackage{booktabs}
\usepackage{url}
\usepackage{hyperref}
\usepackage[linesnumbered,ruled,vlined]{algorithm2e}

\makeatletter
\newcommand*{\defeq}{\mathrel{%
		\rlap{\raisebox{0.3ex}{$\m@th\cdot$}}%
		\raisebox{-0.3ex}{$\m@th\cdot$}}%
	=}\makeatother
\makeatletter
\newcommand*{\eqdef}{=\mathrel{%
		\raisebox{0.3ex}{$\m@th\cdot$}%
		\llap{\raisebox{-0.3ex}{$\m@th\cdot$}}}%
}\makeatother
\makeatletter
\newcommand*{\defeqv}{\mathrel{%
		\rlap{\raisebox{0.3ex}{$\m@th\cdot$}}%
		\raisebox{-0.3ex}{$\m@th\cdot$}}%
	\Longleftrightarrow}\makeatother
\makeatletter
\newcommand*{\eqvdef}{\Longleftrightarrow\mathrel{%
		\raisebox{0.3ex}{$\m@th\cdot$}%
		\llap{\raisebox{-0.3ex}{$\m@th\cdot$}}}%
}\makeatother

\makeatletter
\newcommand*{\defdefeq}{\mathrel{%
		\rlap{\raisebox{0.3ex}{$\m@th\cdot\cdot$}}%
		\raisebox{-0.3ex}{$\m@th\cdot\cdot$}}%
	=}\makeatother

\newcommand\classFont[1]{\textnormal{#1}}
\newcommand{\prob}[1]{\textnormal{\textsc{#1}}}

\newcommand{\sat}{\protect\ensuremath{\prob{Sat}}\xspace}
\newcommand{\cnf}{\protect\ensuremath{\prob{CNF}}\xspace}

\newcommand{\scnf}{\protect\ensuremath{\Sigma_1\prob{-CNF}}\xspace}

\newcommand{\esat}{\protect\ensuremath{\prob{E-Sat}}\xspace}
\newcommand{\emaxsat}{\protect\ensuremath{\prob{E-MaxSat}}\xspace}
\newcommand{\eminsat}{\protect\ensuremath{\prob{E-MinSat}}\xspace}
\newcommand{\ecardmaxsat}{\protect\ensuremath{\prob{E-CMaxSat}}\xspace}
\newcommand{\ecardminsat}{\protect\ensuremath{\prob{E-CMinSat}}\xspace}
\newcommand{\cardminsat}{\protect\ensuremath{\prob{CMinSat}}\xspace}

\newcommand{\ExtTeam}{\protect\ensuremath{\prob{ExtendTeam}}\xspace}
\newcommand{\VerTeam}{\protect\ensuremath{\prob{VerifyTeam}}\xspace}
\newcommand{\ExtMaxTeam}{\protect\ensuremath{\prob{ExtendMaxTeam}}\xspace}
\newcommand{\ExtCardTeam}{\protect\ensuremath{\prob{ExtendCMaxTeam}}\xspace}
\newcommand{\MaxSubTeam}{\protect\ensuremath{\prob{MaxSubTeam}}\xspace}
\newcommand{\mzdh}{\protect\ensuremath{\prob{MaxZerosDualHorn}}\xspace}
\newcommand{\is}{\protect\ensuremath{\prob{IS}}\xspace}
\newcommand{\dhorn}{\protect\ensuremath{\prob{DualHorn}}\xspace}

\newcommand{\dep}[1]{\protect\ensuremath{=\!\!(#1)}\xspace}

\newcommand{\dfn}{\mathrel{=_{\mathrm{def}}}}

\newcommand{\ft}{\protect\ensuremath{\mathbb{X}}\xspace}

\newcommand{\D}{\textnormal{D}}

\newcommand{\Del}{\protect\ensuremath{\classFont{Del}}}
\newcommand{\Inc}{\protect\ensuremath{\classFont{Inc}}}
\renewcommand{\P}{\protect\ensuremath{\classFont{P}}\xspace}
\newcommand{\NP}{\protect\ensuremath{\classFont{NP}}\xspace}

\newcommand{\FO}{\protect\ensuremath{\classFont{FO}}\xspace}

\newcommand{\tu}[1]{\overline{#1}}

\newcommand{\powerset}[1]{\mathcal P\left(#1\right)}

\newcommand{\ie}{i.e.\@\xspace}
\newcommand{\eg}{e.g.\@\xspace}

\newcommand{\ST}{such that\@\xspace}

\newcommand{\ifff}{if and only if\@\xspace}
\newcommand{\resp}{respectively\xspace}

\newcommand{\rel}{\mathrm{rel}}
\newcommand{\team}{\mathrm{team}}

\newcommand{\calA}{\mathcal{A}}

\newcommand{\problemdef}[3]{%
	\begin{center}
		\begin{tabularx}{.95\linewidth}{rX}\toprule
			\textbf{Problem}:& #1\\
			\textbf{Input}:& #2\\
			\textbf{Output}:& #3\\
			\bottomrule
		\end{tabularx}
	\end{center}
}

\newcommand{\problemdefdec}[3]{%
	\begin{center}
		\begin{tabularx}{.95\linewidth}{rX}\toprule
			\textbf{Problem}:& #1\\
			\textbf{Input}:& #2\\
			\textbf{Question}:& #3\\
			\bottomrule
		\end{tabularx}
	\end{center}
}

\newcommand{\free}[1]{\text{free}(#1)}
\usepackage{microtype}

\title{Enumerating Teams in First-Order Team Logics}
\titlerunning{Enumerating Teams in First-Order Team Logics}

\author{Anselm Haak}{Leibniz Universität Hannover, Institut für Theoretische Informatik, Hannover, Germany}{haak@thi.uni-hannover.de}{https://orcid.org/0000-0003-1031-5922}{}
\author{Arne Meier}{Leibniz Universität Hannover, Institut für Theoretische Informatik, Hannover, Germany}{meier@thi.uni-hannover.de}{https://orcid.org/0000-0002-8061-5376}{Funded by the German Research Foundation (DFG), project ME4279/1-2}
\author{Fabian Müller}{Leibniz Universität Hannover, Institut für Theoretische Informatik, Hannover, Germany}{fabian.mueller@thi.uni-hannover.de}{www.thi.uni-hannover.de}{}
\author{Heribert Vollmer}{Leibniz Universität Hannover, Institut für Theoretische Informatik, Hannover, Germany}{vollmer@thi.uni-hannover.de}{https://orcid.org/0000-0002-9292-1960}{}

\authorrunning{A. Haak, A. Meier, F. Müller, H. Vollmer}

\Copyright{Anselm Haak, Arne Meier, Fabian Müller, Heribert Vollmer}

\ccsdesc[100]{F.1.1 Models of Computation}
\ccsdesc[100]{F.1.3 Complexity Measures and Classes}
\ccsdesc[100]{F.4.1 Mathematical Logic}

\keywords{team-based logics, enumeration problem, polynomial delay}

\acknowledgements{We thank anonymous reviewers from MFCS 2020 for pointing out an error in a  result from a previous version.}

\begin{document}

\maketitle

\begin{abstract}
  We start the study of the enumeration complexity of different satisfiability problems in first-order team logics.
  Since many of our problems go beyond DelP, we use a framework for hard enumeration analogous to the polynomial hierarchy, which was recently introduced by Creignou et al. (Discret. Appl. Math. 2019).
  We show that the problem to enumerate all satisfying teams of a fixed formula in a given first-order structure is DelNP-complete for certain formulas of dependence logic and independence logic.
  For inclusion logic formulas, this problem is even in DelP.
  Furthermore, we study the variants of this problems where only maximal, minimal, maximum and minimum solutions, respectively, are considered.
  For the most part these share the same complexity as the original problem.
  An exception is the minimum-variant for inclusion logic, which is DelNP-complete.
\end{abstract}

\section{Introduction}
Decision problems in general ask for the existence of a solution to some problem instance.
In contrast, for \emph{enumeration problems} we aim at generating all solutions.
For many---or maybe most---real-world tasks, enumeration is therefore more natural or practical to study; we only have to think of the domain of databases where the user is interested in all answer tuples to a database query.
Other application areas include web search engines, data mining, web mining, bioinformatics and computational linguistics.
From a theoretical point of view, maybe the most important problem is that of enumerating all satisfying assignments of a given propositional formula. 

Clearly, even simple enumeration problems may produce a big output.
The number of satisfying assignments of a formula can be exponential in the length of the formula.
In \cite{JPY1988}, different notions of efficiency for enumeration problems were first proposed, the most important probably being $\Del\P$ (``polynomial delay''), consisting of those enumeration problems where, for a given instance $x$, the time between outputting any two consecutive solutions as well as pre- and postcomputation times (see \cite{DBLP:conf/foiks/MeierR18}) are polynomially bounded in $|x|$.
Another notion of tractability is captured by the class $\Inc\P$ where the delay and post-computation time can also depend on the number of solutions that were already output.
The separation $\Del\P \subsetneq \Inc\P$ was mentioned in \cite{Strozecki2010}, although one should note that slighlty different definitions were used there.
Several examples of membership results for tractable classes can be found in \cite{DBLP:journals/jcss/LucchesiO78,DBLP:journals/tods/KimelfeldK14,DBLP:journals/fuin/CreignouV15,DBLP:conf/pods/CarmeliKK17,DBLP:conf/csl/BaganDG07,DBLP:conf/pods/DurandSS14}.
As a notion of higher complexity, recently an analogue of the polynomial hierarchy for enumeration problems has been introduced \cite{DBLP:journals/dam/CreignouKPSV19}.
Lower bounds for enumeration problems are obtained by proving hardness (under a suitable reducibility notion) in a level $\Sigma_k^p$ of that hierarchy for some $k \geq 1$ and are regarded as evidence for intractability.

Here, we consider enumeration tasks for so-called \emph{team-based logics}, where first-order formulas with free variables are evaluated in a given structure not for a single assignment to these variables but for sets of such assignments; these sets are called \emph{teams}.
The logical language is extended by so-called generalised dependency atoms (sometimes referred to as \emph{team atoms}) that allow to specify properties of teams, e.g., that the value of a variable functionally depends on some other variable(s) (the dependence atom $\dep{\dots}$~\cite{DBLP:books/daglib/0030191}), that a variable is independent of some other variable(s) (the independence atom $\perp$~\cite{DBLP:journals/sLogica/GradelV13}), or that the values of a variable occur as values of some other variable(s) (the inclusion atom $\subseteq$~\cite{DBLP:journals/apal/Galliani12}). 
Team-based logics were introduced by Jouko Väänänen \cite{DBLP:books/daglib/0030191} and have been used for the study of various dependence and independence concepts important in many areas such as database theory and Bayesian networks (see, e.g., the articles in the textbook by Abramsky et~al.~\cite{DBLP:books/daglib/0037838}). 

For a fixed first-order formula and a given input structure, the complexity of the problem of counting all satisfying teams has been studied by Haak et~al.~\cite{DBLP:conf/mfcs/HaakKMVY19}, where completeness for classes such as $\#\cdot\P$ and $\#\cdot\NP$ was obtained.
In the enumeration context, and in analogy to the case of classical propositional logic as above, it is now natural to ask for algorithms to enumerate all satisfying teams of a fixed formula in a given input structure.
Enumerating teams for formulas with the above mentioned dependency atom thus means enumerating all sets of tuples in a relational database that fulfil the given Boolean combination of \FO-statements and functional dependencies.
In this paper, we consider this problem and initiate the study of enumeration complexity for team based logics.
Notice that, the task of enumerating teams has been considered before in the propositional setting by Meier and Reinbold~\cite{DBLP:conf/foiks/MeierR18}.
We consider team-based logics with the inclusion, the dependence and the independence atom, and study the problems of enumerating all satisfying teams or certain optimal satisfying teams, where optimal can mean maximal or minimal with respect to inclusion or cardinality.
Our results are summarised in Table~\ref{tab:summary} on p.~\pageref{tab:summary}. 
It is known that in terms of expressive power dependence logic corresponds to the class $\NP$.
Hence one cannot expect efficient algorithms for enumerating teams, and in fact, we prove that the problem is $\Del\NP$-complete (\ie, $\Del\Sigma_1^p$-complete) in all but one variants (enumerating all or optimal satisfying teams). For the remaining variant---enumerating inclusion maximal satisfying teams---we show $\Del\NP$-hardness and sketch $\Del\Sigma^p_2$ membership in the conclusion, the precise complexity remains open.
Analogous results hold for independence logic.
Inclusion logic, however, in a model-theoretic sense is equal to the class $\P$ (at least in so-called \emph{lax semantics} \cite{DBLP:conf/csl/GallianiH13}).
Consequently, inclusion logic is less expressive than dependence logic (under the assumption $\P\neq\NP$), and the picture in the enumeration context reflects this: 
We prove that for each inclusion logic formula, there is a polynomial-delay algorithm for enumerating all satisfying teams in a given structure.
This is also true when we want to enumerate all maximal, minimal, or maximum satisfying teams.
Interestingly, enumerating minimum satisfying teams is $\Del\NP$-complete, as for the other logics we consider.

In the next section, we introduce team semantics and the relevant logics.
There, we also introduce algorithmic enumeration and the needed complexity classes, and we formally define the enumeration problems we want to classify in this paper.
In Sect.~\ref{effenum}, we present an efficient enumeration algorithms for inclusion logic, while Sect.~\ref{hardenum} is devoted to the presentation of our completeness proofs for the class $\Del\NP$.
Finally, we summarise our results and conclude with some open questions.
Due to space restrictions, most proofs are only sketched in the paper, but all full details can be found in the appendix.

\section{Definitions and Preliminaries}
We assume familiarity with basic notations from complexity theory \cite{DBLP:books/daglib/0092426}.
We will make use of the complexity classes $\P$ and $\NP$.

\subsection{Team logic}
A \emph{vocabulary} $\sigma=\{\, R_1^{j_1},\dots, R_k^{j_k}\, \}$ is a finite set of relations with corresponding arities $j_1, \dots j_k \in \mathbb N_+$.
A \emph{$\sigma$-structure} $\calA=(A,(R_i^\calA)_{R_i\in\sigma})$ consists of a \emph{universe} $A$ that is a set, and an interpretation of the relations of $\sigma$ in $A$, \ie, $R_i^\calA\subseteq A^{j_i}$ for each $R_i\in\sigma$.
Let $D$ be a finite set of first-order variables and $A$ be some set. 
An \emph{assignment} $s\colon D\to A$ is a function over \emph{domain} $D$ and \emph{codomain} $A$.
The algorithms that we construct later assume an arbitrary order on assignments and thereby on singleton teams. 
For our purposes a lexicographical order suffices.
Moreover, if $s\leq t$ and there exists a $1\leq j\leq n$ such that $s(x_j)<t(x_j)$ then we write $s<t$.

Given an assignment $s$, a variable $x$ and an element $a$ from $A$, the assignment $s(a/x)\colon D\cup\{x\}\to A$ is defined by $s(a/x)(x)\dfn a$ and $s(a/x)(y)\dfn s(y)$ for $x\neq y$.
We call $s(a/x)$ a \emph{supplementing function}.
A \emph{team} is a finite set of assignments with common domain and codomain.
For a team $X$, let $\max(X)$ be the largest assignment contained in $X$ with respect to the lexicographical order on assignments defined before.

Considering a team $X$, a finite set $A$, and a function $F\colon X\to\powerset{A}\setminus\{\emptyset\}$, we then define $X[A/x]$ as the modified team $\{\,s(a/x)\mid s\in X,a\in A\,\}$.
Furthermore, we denote by $X[F/x]$ the team $\{\,s(a/x)\mid s\in X,a\in F(s)\,\}$.
If $X$ is a team whose codomain is the universe of a $\sigma$-structure $\calA$, we say $X$ is a team of $\calA$. 

Now, we proceed with the definition of syntax and semantics of first-order team logic.
Let $\sigma$ be a vocabulary.
Then, the syntax of \emph{first-order team logic}, $\FO[\sigma]$, is defined by the following grammar:
\begin{equation}
\label{eq:fosynt}
\varphi ::= x=y\mid 
x\neq y\mid
R(\tu x)\mid
\lnot R(\tu x)\mid
(\varphi \land \varphi)\mid
(\varphi \lor \varphi)\mid
\exists x.\varphi\mid
\forall x.\varphi,
\tag{$\star$}
\end{equation}
where $\tu x$ is a tuple of first-order variables, $x,y$ are first-order variables, and $R\in\sigma$.
Notice that we restricted the syntax to atomic negation.
The reason for that restriction is the high complexity of problems on formulas with arbitrary negation symbols both in first-order as well as propositional logic \cite{DBLP:books/daglib/0030191,DBLP:journals/tocl/HannulaKVV18}. 

\begin{definition}[Team semantics]
	Let $\sigma$ be a vocabulary, $\calA$ be a $\sigma$-structure, $X$ be a team of $\calA$, $x,y$ be first-order variables, $\tu x$ be a tuple of first-order variables, $R$ be a relation symbol, and $\varphi,\psi\in\FO(\sigma)$.
	The satisfaction relation $\models_X$ for $\FO[\sigma]$-formulas is defined as:
	\[
	\begin{array}{llcl}
	\calA \models_X& x=y &\Leftrightarrow& \forall s\in X\text{ we have that }s(x)=s(y),\\
	\calA \models_X& x\neq y &\Leftrightarrow& \forall s\in X\text{ we have that }s(x)\neq s(y),\\
	\calA \models_X& R(\tu x) &\Leftrightarrow& \forall s\in X\text{ we have that }s(\tu x)\in R^\calA,\\
	\calA \models_X& \lnot R(\tu x) &\Leftrightarrow& \forall s\in X\text{ we have that }s(\tu x)\notin R^\calA,\\
	\calA \models_X& (\varphi\land\psi) &\Leftrightarrow& \calA\models_X\varphi\text{ and }\calA\models_X\psi,\\
	\calA \models_X& (\varphi\lor\psi) &\Leftrightarrow& \exists Y,Z\subseteq X\text{ with } Y\cup Z=X\text{ and }\calA\models_Y\varphi\text{ and }\calA\models_Z\psi,\\
	\calA\models_X&\forall x.\varphi &\Leftrightarrow& \calA\models_{X[A/x]}\varphi,\\
	\calA\models_X&\exists x.\varphi &\Leftrightarrow& \calA\models_{X[F/x]}\varphi\text{ for some }F\colon X\to\powerset{A}\setminus\{\,\emptyset\,\}.
	\end{array}
	\]
\end{definition}

If the underlying vocabulary is clear from the context or not relevant, we usually omit $\sigma$ the expression $\FO[\sigma]$ and write $\FO$ instead.
Let $\varphi\in\FO$ be a first-order team logic formula.
We denote by $\free{\varphi}$ the set of \emph{free variables} in $\varphi$.
Observe that on singletons, the semantics of $φ_1 \lor φ_2$ resemble that of the classical disjunction.
On teams, however, this generalises to the so-called \emph{split junction} operator which literally splits the team into (not necessarily disjunct) parts where each of the formulas $φ_1$ and $φ_2$ has to be satisfied by one of the parts.
Notice that the previously defined semantics are called \emph{lax semantics}.
Furthermore, observe that the empty team satisfies any formula. 
This yields the desirable \emph{flatness} property (a team satisfies a formula \ifff every assignment/singleton from the team satisfies the formula). Note that for a fixed formula $\varphi$ and a given structure $\calA$ there are $\text{dom}(\calA)^{|\free\varphi|}$ different assignments, i.e. a polynomial number of assignments. Since each team is a set of assignments, the size of a team is polynomially bounded as well. Formulae of $\FO(\dep{\dots})$ are {\em closed downwards}, i.e., $\calA\models_X\varphi$ and $Y\subseteq X$ implies $\calA\models_Y\varphi$, formulae of $\FO(\subseteq)$ are {\em closed under unions}, i.e., $\calA\models_X\varphi$ and $\calA\models_Y\varphi$ implies $\calA\models_{X\cup Y}\varphi$ \cite{DBLP:books/daglib/0030191,DBLP:journals/apal/Galliani12}.

\begin{example}
	Consider the formula $\varphi\dfn R(x,y)\lor\lnot R(x,y)$, the structure $\mathcal{A}$ with $R^{\mathcal{A}}=\{\,(0,1),(1,0)\,\}$ and the team $X=\{\,s_1,s_2\,\}$ defined with $s_1(x)=0, s_1(y)=1$, and $s_2(x)=1=s_2(y)$.
	Then $\mathcal{A}\models_X \varphi$ as we can split $X$ into $X_1=\{\,s_1\,\}$ and $X_2=\{\,s_2\,\}$ such that $\mathcal{A}\models_{X_1} R(x,y)$ and $\mathcal{A}\models_{X_2}\lnot R(x,y)$.
\end{example}

Additionally to the connectives defined in the $\FO$-syntax above, we will make use of so-called \emph{generalised dependency atoms}.
We will use the \emph{dependence atom} $\dep{\tu x, \tu  Dy}$, the \emph{inclusion atom} $\tu x \subseteq \tu y$ and the \emph{independence atom} $\tu x \perp_{\tu z} \tu y$ where $\tu x, \tu y, \tu z$ are tuples of first-order variables and $y$ is a first-order variable.
Now for any subset $A \subseteq \{\,\dep{\dots}, \perp, \subseteq\,\}$, we define $\FO(A)$ as first-order logic extended by the respective atoms.
More precisely, we extend the grammar (\ref{eq:fosynt}) by adding a rule for each atom in $A$. For example, for $\FO(\{\,\subseteq\,\})$ we add the rule $φ ::= \tu x \subseteq \tu y$ for any tuples $\tu x, \tu y$ of \FO-variables.
For convenience, we often omit the curly brackets and write for example $\FO(\subseteq)$ instead of $\FO(\{\,\subseteq\,\})$.
The logics $\FO(\dep{\dots})$, $\FO(\subseteq)$ and $\FO(\perp)$ are called \emph{dependence logic}, \emph{inclusion logic} and \emph{independence logic}, \resp.

Intuitively, an independence atom expresses that two tuples are independent with respect to a third tuple.
A tuple $\tu x$ depends on another tuple $\tu y$, so $\dep{\tu x,\tu y}$, if for every pair of assignments from the team that agree on $\tu x$ also agree on $\tu y$. 
This is the idea of functional dependency in the database setting.
A tuple $\tu x$ is included in a tuple $\tu y$, that is $\tu x\subseteq\tu y$, if for every assignment $t_1$ in the team there exists another one $t_2$ such that $\tu x$ under $t_1$ coincides with $\tu y$ under $t_2$.
Before we formally define the semantics for these three atoms, we need to introduce a little bit of notation.
If $\tu x=(x_1,\dots,x_n)$ is a tuple of first-order variables for $n\in\mathbb N$, and $s$ is an assignment, then $s(\tu x)\dfn (s(x_1),\dots,s(x_n))$.

\begin{definition}[Generalised dependency atoms semantics]
	Let $\sigma$ be a vocabulary, $\calA$ be a $\sigma$-structure, $X$ be a team of $\calA$, and $\tu x, \tu y, \tu z$ be tuples of first-order variables.
	The satisfaction relation $\models_X$ for $\FO(\sigma)$-formulas then is extended as follows:
	\[
	\begin{array}{llcl}
	\calA \models_X& \tu x\perp_{\tu z}\tu y &\Leftrightarrow& 
	\forall s,t\in X\text{ with }s(\tu z)=t(\tu z)\;\;
	\exists u\in X\text{ such that }\\&&&
	u(\tu x)=s(\tu x),\; 
	u(\tu y) =t(\tu y),\;
		u(\tu z)=s(\tu z).\\
	\calA \models_X& \tu x\subseteq\tu y &\Leftrightarrow& 
	\forall s\in X\;\exists t\in X\text{ such that }s(\tu x)=t(\tu y).\\
	\calA \models_X& \dep{\tu x,\tu y} &\Leftrightarrow& 
	\forall s,t\in X\text{ we have that }s(\tu x)=t(\tu x)\text{ implies }s(\tu y)=t(\tu y).
	\end{array}
	\]
\end{definition}

In the following, we define the model checking problem on the level of first-order team logic formulas in the setting of data complexity (fixed formula).
\problemdefdec{$\VerTeam_\varphi$}{Structure $\calA, \text{ team }X$}{$\calA \models_X \varphi \text{ and } X \neq \emptyset?$}
	
\begin{lemma}\label{vernp}
	Let $A\subseteq\{\,\perp,\subseteq,\dep{\dots}\,\}$, $\varphi\in\FO(A)$. Then $\VerTeam_\varphi\in\NP$.
\end{lemma}
\begin{proof}
Every fixed formula is of bounded width (width is the maximal number of free variables in subformulas of a given formula).
As all of the generalised dependency atoms in $A$ can be evaluated in polynomial time, a result from Grädel~\cite[Theorem 5.1]{DBLP:journals/tcs/Gradel13} applies, yielding $\VerTeam_\varphi\in\NP$.
\end{proof}

Our algorithms often start with either $\emptyset$ or $\text{dom}(\calA)^{|\free\varphi|}$ (the full team) as one of their inputs, for a fixed formula $\varphi$ and a structure $\calA$. Instead of $\text{dom}(\calA)^{|\free\varphi|}$ we will write $\ft.$

The following proposition summarises important results from literature that are referenced later in proofs.
It mainly states key connection between team logics and predicate logic, also mentioning descriptive complexity results that are consequences of these connections.

\begin{proposition}[\cite{DBLP:journals/apal/Galliani12,DBLP:journals/jolli/KontinenV09,DBLP:conf/csl/GallianiH13}]\label{ind2sigma11}\label{expr}
	\ \\\vspace{-8pt}
	\begin{enumerate}
		\item\label{propit:dep=ind=np}
      Over sentences both $\FO(\perp)$ and $\FO(\dep{\dots})$ are expressively equivalent to $\Sigma^1_1$: Every $\sigma$-sentence of $\FO(\perp)$ (or $\FO(\dep{\dots})$) is equivalent to a $\sigma$-sentence $\psi$ of $\Sigma^1_1$, i.e., for any $\sigma$-structure $\calA$,
		  $\calA\models\varphi\iff\calA\models\psi,$
		  and vice versa. As a consequence of Fagin's Theorem \cite{fagingeneralized}, over finite structures both $\FO(\perp)$ and $\FO(\dep{\dots})$ capture \NP.
		\item\label{myopic} 
		  Let $\varphi(R)$ be a \emph{myopic} $\sigma$-formula, that is, $\varphi(R)=\forall \tu{x} (R(\tu{x})\rightarrow \psi(R,\tu{x}))$, where $\psi$ is a first order $\sigma$-formula with only positive occurrences of $R$. Then there exists a $\sigma$-formula $\chi \in \FO(\subseteq)$ such that for all $\sigma$-structures $\mathcal{A}$ and all teams $X$ we have $\mathcal{A} \models_X \chi(\tu{x}) \Leftrightarrow \mathcal{A}, \textnormal{rel}(X) \models \varphi(R)$.
	\end{enumerate}
\end{proposition}

\subsection{Enumeration}

For the basics of enumeration complexity theory, we follow Creignou et al. \cite{DBLP:journals/dam/CreignouKPSV19}.  

In contrast to decision problems where one gets an input and often has to answer whether there is a ``solution'' to the input, for enumeration problems one has to compute the set of all solutions to the input. As an example see the difference between the decision problem $\sat^{\team}_\varphi$ and the enumeration problem $\esat^{\team}_\varphi$.

	\problemdefdec{$\sat^{\team}_\varphi$}{Structure $\calA$}{$\{\,X \mid \calA \models_{X} \varphi \text{ and } X \neq \emptyset\,\} \neq \emptyset?$}	

	\problemdef{$\esat^{\team}_\varphi$}{Structure $\calA$}{$\{\,X \mid \calA \models_X \varphi \text{ and } X \neq \emptyset\,\}$}

Note that for all our problem definitions, if not otherwise stated, φ is a formula from $\FO(A)$ for some $A \subseteq \{\,\dep{\dots}, \subseteq, \perp\,\}$.

As these sets can get exponentially large compared to the input our, classical measures (like runtime of the machine/algorithm) will not suffice.
To be able to talk about tractability and intractability of problems in the enumeration setting we need to define new classes. 
The idea is that we will not bound the time of the whole computation, but the time of the computations between the outputs of two consecutive solutions, which we will call \emph{delay}. 
Instead of Turing machines we will use random access machines (RAMs), to be able to access the (potentially) exponential ``memory'' in polynomial time.	

\begin{definition}[\cite{DBLP:journals/dam/CreignouKPSV19}]
	Let $C$ be a decision complexity class and $p$ be a polynomial. The enumeration class $\Del C$ consists of all enumeration problems $E$, for which there exists a RAM $M$ with oracle $L\in C$ such that for all inputs $x$, $M$ enumerates the output set of $E$ with $p(|x|)$ delay and all oracle queries are bounded by $p(|x|).$
\end{definition}

\begin{example} We show $\esat \in \Del\NP$. 
	\problemdef{$\esat$}{Propositional forumlua $\varphi$}{$\{\,\beta \mid \beta \models \varphi\,\}$}
	Let $\varphi$ be our input formula over the variables $x_1, \dots, x_n$. We start by assigning the value $0$ to variable $x_1$ and ask the oracle $\sat$ (satisfiability of propositional formulas) if the resulting formula is satisfiable. If the answer is ``no'', we know that, there is no satisfying assignment for $\varphi$, which assigns the value $0$ to variable $x_1$ and we therefore ask the oracle again but this time we assign the value $1$ to variable $x_1$. If the answer is ``yes'' we continue by assigning the value $0$ to variable $x_2$ and ask our oracle again. That means for each ``yes'' we go one step down in the tree of assignments and assign the value $0$ to the next variable, if the answer is ``no'' and we did not assign the value $1$ to the current variable before then we assign the value $1$ to it this time and if the answer is ``no'' and we assigned the value $1$ to the current variable before, we go one step up in the tree off assignments. If at some point we assigned all variables and get the answer ``yes'', we output the current (satisfying) assignment. If we gone through all assignments this way we output that there is no further satisfying assignment and halt.

We now have to argue, that this method has polynomial delay, the oracles questions are polynomially bounded and that the oracle is in $\NP$. The last one is the easiest, since we all know $\sat \in \NP.$ The oracles questions have the same length as the input formula, therefore they are polynomial bounded. To get from one satisfying assignment to another we have to go up and down the whole tree of assignments once in the worst case. Since the depth is $n$ this takes $p(n)$ time.

The method is called flashlight or torchlight search, and our algorithms in sections~\ref{effenum} and \ref{hardenum} showing membership for $\Del\P$ and $\Del\NP$ will be based on it.

\end{example}

To be able to show hardness for our new classes we need a suitable definition of reducibility.
The reduction we use is quite similar to a Turing reduction in the decision case.
For this we give a machine access to an enumeration oracle to solve another enumeration problem.
The kind of machine we use here is called \emph{enumeration oracle machine} (EOM) which is a RAM with some new special registers: an infinite number of registers for the oracle questions and one register for the answer.
The machine can write an oracle question into the respective registers (one bit per register) and in one step the answer appears in the register for the answer.
If there are further solutions to the question that were not given before, the answer is a solution.
Otherwise, the answer is a special symbol, meaning that all solutions have been given.
The machines that we use are also \emph{oracle-bounded}, that is, all oracle questions are polynomial in the size of the input. 

\begin{definition}[\cite{DBLP:journals/dam/CreignouKPSV19}]
	Let $E_1,E_2$ be enumeration problems.
  We say that $E_1$ reduces to $E_2$ via $D$-reductions, $E_1 \le_{\D} E_2$, if there is an oracle-bounded EOM $M$ that enumerates $E_1$ using oracle $E_2$ with polynomial delay and independently of the order in which the $E_2$-oracle enumerates it answers.
\end{definition}

\begin{proposition}[\cite{DBLP:journals/dam/CreignouKPSV19}]
  The class $\Del \Sigma_k^p$ is closed under $D$-reductions for any $k \in \mathbb{N}$.
\end{proposition}

Let $\prob{E}$ be the enumeration problem, given input $x$, to output the set of solutions $S(x)$.
We denote by $\prob{Exist-E}$ the problem to decide, given $x$, whether $|S(x)| \geq 1$.

\begin{proposition}[\cite{DBLP:journals/dam/CreignouKPSV19}]\label{existhard}
	Let $E$ be an enumeration problem and $k \ge 1$ such that $\prob{Exist-E}$ is $\Sigma_k^p$-hard. 
	Then we have that $E$ is $\Del\Sigma_k^p$-hard under $D$-reductions.
\end{proposition}

We slightly generalise this theorem:

\begin{theorem}\label{arbhard}
	Let $A$ be an $\Sigma_k^p$-hard decision problem and let $E$ be an enumeration problem such that $A$ can be decided in polynomial time by an algorithm that has access to oracle $E$. Then it holds that $E$ is $\Del\Sigma_k^p$-hard under $D$-reductions.
\end{theorem}

\begin{proof}
	The proof is essentially the same as the one for Prop.~\ref{existhard}.
	Let $B \in \Del\Sigma_k^p$ and $L\in\Sigma_k^p$ be a witness for $B \in \Del\Sigma_k^p$, that is, there is an algorithm with access to oracle $L$ that enumerates $B$ with polynomial delay. Since $A$ is $\Sigma_k^p$-hard and by the precondition of the theorem ($A$ can be decided in polynomial time by an algorithm with an $E$-oracle), we can answer the oracle questions to $L$ by asking $E$ instead. 
	It follows that $B$ can be enumerated by an algorithm with an $E$-oracle with polynomial delay. 
\end{proof}

We will close this subsection defining four more enumeration problems.
In the following two sections we analyse the complexity of the defined problems for our different logics.

\problemdef{$\emaxsat^{\team}_\varphi$}{Structure $\calA$}{$\{\,X \mid \calA \models_X \varphi, X \neq \emptyset \text{ and } \forall X' \ X \subsetneq X' \Rightarrow \calA \not\models_{X'} \varphi\,\}$}


\problemdef{$\ecardmaxsat^{\team}_\varphi$}{Structure $\calA$}{$\{\,X \mid \calA \models_X \varphi, X \neq \emptyset \text{ and } \forall X' \ |X'|> |X| \Rightarrow \calA \not\models_{X'}\varphi \,\}$}

\noindent The dual problems $\eminsat^{\team}_φ$ and $\ecardminsat^{\team}_φ$ require the conditions $\forall X'\neq \emptyset \ X' \subsetneq X \Rightarrow \calA \not\models_{X'} \varphi$ and $\forall X'\neq \emptyset \ |X'|< |X| \Rightarrow \calA \not\models_{X'}\varphi$ instead, respectively.


\section{Efficient Enumeration}\label{effenum}

In this section, we study the class $\Del\P$. All the results are for inclusion logic and rely on the fact that $\MaxSubTeam$---the problem to compute the maximal subteam of a given team satisfying a given inclusion logic formula in a given structure---is computable in polynomial time. This was shown for modal propositional inclusion logic \cite{DBLP:journals/logcom/HellaKMV19}. Our case can be proven similar by induction. Usually this result is not usable for satisfiability since one has to give $\MaxSubTeam$ the full team $\mathbb{X}$ which is exponentially large compared to a given formula, but since we fix the formula this is not a problem.

 Note that for inclusion logic the maximal satisfying team is unambiguous: if there are two satisfying teams $X, X'$ of same size, then $X \cup X'$ is also satisfying due to union closure. The teams $X, X'$ therefore can not be maximal with respect to cardinality and inclusion. 

\problemdef{$\MaxSubTeam_\varphi$}{Structure $\calA, \text{ team }X$}{$X'$ with $\calA \models_{X'} \varphi, X' \subseteq X \text{ and } \forall X''\subseteq X \colon |X''|>|X'| \Rightarrow \calA\not\models_{X''}\varphi $}

In our algorithms we use $\MaxSubTeam_\varphi$ as an oracle, but one could also call it as a subroutine, since $\Del\P^\P=\Del\P.$ 

\begin{theorem}\label{inclP}
	For any formula $\varphi \in \FO(\subseteq)$ it holds that $\esat^{\team}_\varphi \in \Del\P.$
\end{theorem}

\begin{proof}		
	We construct a recursive algorithm with access to a $\MaxSubTeam$ oracle that on input $(\calA,X,Y)$ enumerates all satisfying subteams $X' \neq \emptyset$ of $X$ with $Y\subseteq X'$. To compute for a given $\calA$ all satisfying subteams, we then need to run this algorithm on input $(\calA, \ft, \emptyset)$. 

	\begin{algorithm}[H]
	\DontPrintSemicolon
  \label{alg:esatinc}
	\caption{Algorithm used to show $\esat^{\team}_φ \in \Del\P$ for $φ \in \FO(\subseteq)$}
    \SetKwProg{Fn}{Function}{}{end}

    \Fn{\textnormal{EnumerateSubteams}(structure $\calA$, teams $X, Y$) with oracle $\MaxSubTeam$}	{
		  $X\gets\MaxSubTeam_\varphi(\calA,X)$\;
		  \If{$X\neq \emptyset \wedge Y \subseteq X$}{
		  	output $X$\;
		  	\For{$s\in X$}{
		  		$Y = \{\,s' \mid s' < s \wedge s'\in X \,\} $ \;
		  		EnumerateSubteams$(\calA,X\setminus\{\,s\,\},Y)$}}
    }
	\end{algorithm}
  
  The algorithm does not output any solution more than once.
  In the recursive calls, it only outputs solutions where at least one assignment is omitted from the maximal solution, which is the only solution output before.
  Also, when the assignment $s$ is chosen in the for-loop, the next recursive call only outputs solutions that omit $s$, but contain all assignments $s' < s$ that were present in $X$.
  In contrast, in every solution found in previous recursive calls, at least one of the assignments $s' < s$ from $X$ was omitted.
  On the other hand, the algorithm outputs every solution at least once.
  Every solution is a subset of the maximal satisfying subteam of $X$ and the algorithm starts with that maximal solution and then recursively looks for all strict subsets of it.
  This can be seen by noticing that when choosing the assignment $s$ in the for-loop, the next recursive call outputs all satisfying subteams of $X$ that exclude $s$, except for those that also exclude some $s' < s$ from $X$ and were hence output before.
\end{proof}

\begin{theorem}\label{inclmaxP}
	Let $\varphi \in \FO(\subseteq)$. Then $\eminsat^{\team}_\varphi \in \Del\P.$
\end{theorem}
\begin{proof} 
  This can be proven similar to Theorem~\ref{inclP} by slightly modifying Algorithm~\ref{alg:esatinc} \ST it takes input $(\calA,X,Y)$ and computes all inclusion minimal satisfying subteams $X'\neq\emptyset$ of $X$ with $Y\subseteq X'$.
  The only change needed for this is that it only outputs a team $X$, if \MaxSubTeam answers $\emptyset$ for all $X\setminus\{\,s\,\}$, where $ s\in X$.
  
  \begin{algorithm}
  	\caption{Algorithm used to show $\eminsat^{\team}_φ \in \Del\P$}
  	\label{alg:minsatincl}
  	\SetKwProg{Fn}{Function}{}{end}
  	\DontPrintSemicolon
  	\Fn{\textnormal{EnumerateMinSubteams}($\text{structure } \calA, \text{ teams } X, Y$) with oracle $\MaxSubTeam$}	{
  		$X\gets\text{\MaxSubTeam}(\calA,X)$\;
  		\If{$X\neq \emptyset \wedge Y \subseteq X$}{
  			\lIf{$\forall s \in X \ \textnormal{\MaxSubTeam}(\calA, X \setminus \{\,s\,\})=\emptyset$}{output $X$}\Else{
  				\For{$s\in X$}{
  					$Y = \{\,s' \mid s' < s \wedge s'\in X \,\} $ \;
  					EnumerateMinSubteams$(\calA,X\setminus\{\,s\,\},Y)$}}}
  	}
  \end{algorithm}
  
\end{proof}

The next result follows from the fact, that \MaxSubTeam can be computed in polynomial time, since the solution set only consists of the maximal satisfying team for both problems.

\begin{theorem}
	For $\varphi \in \FO(\subseteq)$ the problems $\emaxsat_\varphi^{\team}, \ecardmaxsat_\varphi^{\team}$ are included in $\Del\P$.
\end{theorem}

Note that there is an enumeration problem we did not mention in this section, which is $\ecardminsat_\varphi^{\team}$. This is due to the fact, that this problem is actually $\Del\NP$-complete as we will see in the next section.    

\section{A Characterisation of DelNP}\label{hardenum}

We show that for certain formulas the problem $\esat_\varphi^{\team}$ captures the class $\Del\NP$. 
Moreover, we will extend this result to all remaining cases, that is, all combinations of logics and problems we did not classify already in Section~\ref{effenum}.
\begin{theorem}\label{sathard}
	Let $A\subseteq \{\,\dep{\dots},\perp\,\}, A \neq \emptyset$. 
	There exists a formula $\varphi \in \FO(A)$ such that the problem $\sat^{\team}_\varphi$ is $\NP$-hard.
\end{theorem}

\begin{proof}
  We show the result for $A = \{\,\perp\,\}$. The proof for $A = \{\,\dep{\dots}\,\}$ works analogously by reducing from the $\NP$-complete problem $\scnf^-$, that is, given a propositional formula $φ \in \scnf^-$, decide whether φ is satisfiable.
  Here, $\scnf$ is the class of propositional formulas with existential quantifiers in prenex normal form and where the quantifier-free part is in conjunctive normal form.
  The negative fragment $\scnf^-$ further restricts formulas by allowing free variables to only occur negatively.
  
  We reduce from the $\NP$-complete problem $\cnf\prob{-}\sat$ to the problems $\sat_φ^\rel$ and $\sat_φ^{\rel*}$ for some $φ \in \Sigma_1^1$, see below for formal definitions.
  By Proposition~\hyperref[propit:dep=ind=np]{\ref*{ind2sigma11} item \ref*{propit:dep=ind=np}} we get that $\sat_{\varphi'}^{\team}$ is $\NP$-hard, for a formula $\varphi' \in \FO(\perp)$. Let $\varphi$ be a $\Sigma_1^1$-formula.\medskip

	\problemdefdec{$\sat_\varphi^\rel$}{Structure $\calA$}{$\{\,R \mid \calA, R \models \varphi \,\} \neq \emptyset?$}	
	\problemdefdec{$\sat_\varphi^{\rel*}$}{Structure  $\calA$}{$\{\, R \mid  \calA, R \models \varphi \text{ and } R\neq \emptyset\,\} \neq \emptyset?$}

	Let $\psi(x_1,\dots,x_n)=\bigwedge_i^m C_i$ be a propositional formula in conjunctive normal form, with $C_i=\bigvee_{j} l_{i,j}.$
  We encode $\psi$ via the structure $\calA(ψ)=\{\,\{\,x_1,\dots,x_n,C_1,\dots,C_m\,\}, P^2,N^2\,\}$, where $(C,x) \in P$ ($(C,x) \in N$) \ifff variable $x$ occurs positively (negatively) in clause $C$. We define the following $\Sigma_1^1$-formula $χ(R)$ over vocabulary $(P^2, N^2)$:
  \[\chi(R)= \forall C \ \exists x \ ((P(C,x) \wedge R(x)) \vee (N(C,x) \wedge \neg R(x))).\]
  Now, we have that $\exists R\colon \mathcal{A}(ψ), R \models χ(R) \iff ψ$ is satisfiable, showing $\cnf\prob{-}\sat \le_m^p \sat_\chi^\rel$.
	
  Next, we will show $\NP$-hardness for $\sat^{\rel*}_{\chi'}$.
  This follows from an easy reduction from $\sat^{\rel}_{\varphi}$ to $\sat^{\rel*}_{\varphi'}$ which holds for all $\varphi \in \Sigma_1^1$.
  Let $\varphi'(R)=\varphi(R)\vee \varphi(\emptyset)$.
  Now, for all structures $\calA$ we claim that $\exists R \colon \calA, R\models \varphi(R) \iff \exists R'\neq \emptyset \colon \calA, R'\models \varphi'(R')$.
	
	``$\Rightarrow$'': If $\calA, R \models \varphi(R)$ only holds for $R=\emptyset$, then $\calA, R' \models \varphi'(R')$ holds for any $R'$, in particular for any $R'\neq \emptyset.$
	If $\calA, R \models \varphi(R)$ for any $R\neq\emptyset$, then $\calA, R \models \varphi'(R)$ also holds.
	
	``$\Leftarrow$'': Since $\calA, R \not\models \varphi(R)$ for all $R$, in particular we have $\calA, \emptyset \not\models \varphi(\emptyset)$. 
	This immediately shows $\calA, R \not\models \varphi'(R)$ for all $R$.
\end{proof}

\begin{corollary}\label{delnphard}
	For $A\subseteq \{\,\dep{\dots},\perp\,\}, A \neq \emptyset $ there exists a formula $\varphi\in \FO(A)$ such that the problems $\esat_\varphi^{\team}, \emaxsat_\varphi^{\team}, \ecardmaxsat_\varphi^{\team}, \eminsat_\varphi^{\team}, \ecardminsat_\varphi^{\team}$ are $\Del\NP$-hard.
\end{corollary}	

\begin{proof}
	By Theorem~\ref{sathard}, there is a formula $\varphi \in \FO(A)$ (with $A\subseteq \{\,\dep{\dots},\perp\,\}$) such that $\sat_\varphi^{\team}$ is $\NP$-hard. 
	Since $\sat_\varphi^{\team}$ can be decided in polynomial time by an algorithm with oracle access to any of the problems mentioned in this corollary (simply ask the oracle and return ``no'' \ifff the output is $\bot$), by Theorem~\ref{arbhard}, it follows that all of these problems are $\Del\NP$-hard.
\end{proof}

\begin{theorem}\label{mem}
	For $A=\{\,\perp,\dep{\dots},\subseteq\,\}$ and $\varphi \in \FO(A)$, we have that $\esat^{\team}_\varphi \in \Del\NP.$
\end{theorem}

\begin{proof}
	We give a recursive algorithm enumerating $\esat^{\team}_\varphi$ with polynomial delay, when given oracle access to $\ExtTeam_{\varphi}$ (for definition see below) and $\VerTeam_\varphi$.
	\problemdefdec{$\ExtTeam_\varphi$}{Structure $\calA, \text{ team }X, \text{ set of assignments }Y$}{$\{\,X' \mid \calA \models_{X'} \varphi, X \subsetneq X' \text{ and } X' \cap Y = \emptyset \,\} \neq \emptyset?$}
		
	$\ExtTeam_φ \in \NP$ for all φ: A team $X'$ is guessed and $X \subsetneq X' ∧ X' \cap Y = \emptyset$ can be checked in polynomial time. Finally, $\calA \models_{X'} \varphi$ can be decided in $\NP$ by Lemma~\ref{vernp}.
	
 We now construct an algorithm that gets a structure $\calA$ and a team $X$ as inputs and outputs all satisfying teams $X'$ with $X\subseteq X'$ and $X'\setminus X \subseteq \{\,s \in   \text{dom}(\calA)^k \mid s > \textnormal{max}(X)\,\}$, that is, $X'$ only contains new assignments that are larger than the largest assignment in $X$.
 The algorithm searches these teams $X'$ by using recursive calls where exactly one assignment $s > \max(X)$ is added to $X$.
 By design, the recursive call where $s'$ is added only outputs teams that contain $s'$ and no assignment between $\max(X)$ and $s'$, ensuring that no team is output twice.
 We run the algorithm with input $(\calA, \emptyset)$ to get all satisfying teams.
	
	\begin{algorithm}[H]\label{alg:esat}
	\DontPrintSemicolon
	\caption{Algorithm used to show $\esat^{\team}_φ \in \Del\NP$ for $φ \in \FO(A)$}
    \SetKwProg{Fn}{Function}{}{end}
	
    \Fn{\textnormal{EnumerateSuperteams}(structure $\calA$, team $X$) with oracles $\ExtTeam_{\varphi}$ and $\VerTeam_{\varphi}$}{
		  $Y=\bigcup_{s < \text{max}(X) \wedge s \not\in X} s$\;
		  \lIf{$\VerTeam_\varphi(\calA,X)$}{output $X$}
		  \If{$\ExtTeam_\varphi(\calA,X,Y)$}{
		  	\ForAll{$s > \textnormal{max}(X)$}{
		  		$\text{EnumerateSuperteams}(\calA,X \cup \{\,s\,\})$}}
    }
	\end{algorithm}
\end{proof}

\begin{theorem} \label{memcardmax}
	For $A=\{\,\perp,\dep{\dots},\subseteq\,\}$, $\varphi \in \FO(A)$, we have that $\ecardmaxsat^{\team}_\varphi \in \Del\NP$.
\end{theorem}

\begin{proof}	
	There is a recursive algorithm that on input $(\calA,X,k)$ enumerates all satisfying superteams of $X$ having cardinality $k$ with polynomial delay. 
	The algorithm is very similar to the one used for Theorem~\ref{mem}. The only differences are that $|X|=k$ is checked before a team $X$ is output and that $\ExtCardTeam_{\varphi}$ is used as the oracle instead of $\ExtTeam_{\varphi}$. 
	\problemdef{$\ExtCardTeam_\varphi$}{Structure $\calA, \text{ team }X, \text{ set of assignments } Y, \text{ natural number } k $}{$\{\,X' \mid  \calA \models_{X'} \varphi,  X \subsetneq X' , X' \cap Y = \emptyset \text{ and } |X'|=k\,\} \neq \emptyset  $}
	\begin{algorithm}\label{enumcardmaxteams}
		\caption{Algorithm used to show $\ecardmaxsat^{\team}_\varphi \in \Del\NP$}
		\label{alg:cardmax}
		\SetKwProg{Fn}{Function}{}{end}
		\DontPrintSemicolon
		\Fn{\textnormal{EnumerateCMaxTeams}($\text{structure }\calA,\text{ team }X,\text{ natural number }k$) with oracles $\ExtCardTeam_{\varphi}$ and $\VerTeam_{\varphi}$}{
			$Y=\bigcup_{s < \text{max}(X) \wedge s \not\in X} s$\;
			\lIf{$\VerTeam_\varphi(\calA,X)\wedge |X|=k$}{output $X$}
			\ElseIf{$\ExtCardTeam_\varphi(\calA,X,Y,k)$}{
				\For{$s > \textnormal{max}(X)$}{
					$\textnormal{EnumerateCardMaxTeams}(\calA,X \cup \{\,s\,\},k)$}}}
	\end{algorithm} 
	The maximum cardinality $k$ can be computed by asking the $\ExtCardTeam_\varphi$ oracle on input $(\calA,\emptyset,\emptyset,i)$ for $i=|\text{dom}(\calA)|^{|\free\varphi|}, \dots, 1$.

\end{proof}

\begin{theorem}\label{memmin}
	For $A\subseteq \{\,\perp,\dep{\dots},\subseteq\,\}$ and $\varphi \in \FO(A)$ the problems $\eminsat_\varphi^{\team}$, $\ecardminsat_\varphi^{\team}$ are included in $\Del\NP.$
\end{theorem}
\begin{proof}
	For $\eminsat_\varphi^{\team}$ we can run a slightly modified version of Algorithm~\ref{alg:esat} on input ($\mathcal{A}, \emptyset)$, which was originally used for $\esat_\varphi^{\team}$.
  The only modification needed is that the new algorithm terminates after outputting a solution.
	
	We can solve $\ecardminsat_\varphi^{\team}$ similarly, but this time adjust the algorithm we described in Theorem~\ref{memcardmax}. We compute the minimal $k$ (instead of the maximal) for which $\ExtCardTeam_{\varphi}$ is true before starting the Algorithm with that $k$.
  Also, the new algorithm again terminates after outputting a solution.
\end{proof}

In the next result, we show \NP-hardness for the decision problem $\cardminsat_\varphi^\team$, for an inclusion logic formula $\varphi$. 
	\problemdefdec{$\cardminsat_\varphi^{\team}$}{Structure $\calA, k\in\mathbb{N}$}{$\{\, X \mid \calA \models_{X} \varphi, X \neq \emptyset \text{ and }|X| \le k \,\} \neq \emptyset?$}
By this and Theorem~\ref{arbhard}, we can conclude \Del\NP-hardness for $\ecardminsat_\varphi^\team$. 
We reduce from the \NP-complete problem $\is^*$ (\textsc{IndependentSet}) to $\cardminsat_\varphi^\team$ with two intermediate steps.

\problemdefdec{$\is^*$}{Graph $G=(V,E), k\in\mathbb{N}$}{$\{\, V' \mid  \forall u,v \in V'\colon \{\,i,j\,\} \not\in E, V' \subsetneq V, |V'| \ge k \text{ and } V' \subseteq V\,\} \neq \emptyset?$}

Note that $\is^*$ is \NP-complete: We can reduce from the standard version $\is$, where $V'=V$ is allowed, by just adding one new vertex which is connected to all old vertices. 
The problems remaining problems we need for this reduction are defined as follows.  
\problemdefdec{$\cardminsat_\varphi^{\rel}$ for $φ \in \Sigma_1^1$}{Structure $\calA, k\in\mathbb{N}$}{$\{\, R \mid\calA, R \models \varphi, R \neq \emptyset \text{ and } |R| \le k  \,\} \neq \emptyset?$}

\problemdefdec{$\mzdh^*$}{Propositional dual-horn formula $\varphi, k\in\mathbb{N}$}{$\{\,\beta \mid \beta \models \varphi, \beta \neq \emptyset\text{ and } |\beta| \le k \,\} \neq \emptyset?$}
For this, we represent propositional assignments β by the set (relation) of variables it maps to $1$. 
Also, we call $|β|$ the \emph{weight} of β.\medskip

\begin{theorem}\label{cardhard}
	There is a formula $\varphi \in \FO(\subseteq)$ such that $\cardminsat_\varphi^{\team}$ is $\NP$-hard.
\end{theorem}

\begin{proof}
  We reduce from the \NP-complete problem $\is^*$, showing that there are a myopic formula $φ' \in \Sigma_1^1$ and a formula $φ \in \FO(\subseteq)$ \ST
	\[\is^* \underset{(1)}{\leq^P_m} \mzdh^* \underset{(2)}{\leq^P_m} \cardminsat_{\varphi'}^{\rel} \underset{(3)}{\leq^P_m} \cardminsat_\varphi^{\team}.\]
	
	For (1) an arbitrary $(G=(V,E),k)$ is mapped to $(\varphi=\bigwedge_{\{\,i,j\,\} \in E} x_i \vee x_j, |V|-k)$. Intuitively, assigning a variable $x_i$ to $0$ in φ corresponds to picking the vertex $i$ in $G$ for an independent set. The formula φ expresses that at most one of the variables in any clause may be set to $0$, corresponding to the condition that at most one of the endpoints of an edge can be in an independent set.
  From this it can easily be seen that there is a $1$-$1$-correspondence between indpendent sets $V'$ of $G$ of size at least $k$ and satisfying assignments of φ of weight at most $k$.
  Note that $\varphi$ is obviously a \dhorn formula.
  
  Let $σ = (P^2, N^2)$ be a vocabulary.
  A propositional CNF-formula χ can be encoded as a $σ$-structure $\mathcal{A}_χ$ as follows: The universe contains the variables and clauses of χ.
  The relation $P^{\mathcal{A}_χ}$ ($N^{\mathcal{A}_χ}$) contains a pair $(C,x)$, if $C$ is a clause in χ, $x$ is a variable and $x$ occurs positively (negatively) in $C$ in the formula χ.

  For (2), define the myopic second-order formula $φ'$ over σ as follows:
  \begin{align*}
  	φ'(R) = \forall x \ (R(x) \rightarrow (\forall C\ &((\neg \exists z \ N(C,z))\rightarrow (\exists y \ P(C,y)\wedge R(y))) \\
  	   &\wedge (N(C,x)\rightarrow (\exists y \ P(C,y)\wedge R(y)))))
  \end{align*}
  
  Now suppose $R$ satisfies the formula $\phi'$. Let $x \in R$. It follows that all clauses that contain $x$ or contain only positive literals are satisfied by $R$: If $x$ is positively contained in a clause $C$, then it is already satisfied since $x \in R$. If $x$ is negatively contained in $C$, then there must be another variable $y$ that occurs positively in $C$ (since each clause contains at most one negative literal) with $y \in R$. If $C$ only contains positive literals, then there must be one $y \in R.$ This only works if there is at least one variable included in $R$. If $R$ is empty in the first place the premise of the first implication is always false and therefore the conclusion can be anything. It follows that $\phi'(\emptyset)$ is always true, which is no surprise since it is a myopic formula. But since we are only looking for non-empty relations, non zero-assignments $\beta$ respectively this is not a problem.
	Now for all assignments $\beta\neq \emptyset$ it holds that $\beta \models χ \iff \calA_χ, \beta \models φ'(\beta)$.
		
  Finally, (3) follows from Proposition~\hyperref[myopic]{\ref*{ind2sigma11} item \ref*{myopic}}, since $φ'$ is a myopic formula. 
\end{proof}

The second and third reductions are essentially the same that were used to show $\#\prob{DualHorn} \subseteq \#\FO(\subseteq)$ \cite{DBLP:conf/mfcs/HaakKMVY19}. The difference is that in the counting case, the number of solutions to the \classFont{DualHorn}-formula must be equal to number of solutions to the $\FO(\subseteq)$-formula, and in our case the size of maximal and minimal solutions must preserved. Fortunately the given formula in the second reduction delivers both, as the solutions are exactly the same for both formulas.

Note that this reduction also works if we use positive \classFont{2CNF}-formulas (propositional formula in conjunctive normal form, where each clause has two positive literals) instead of \classFont{DualHorn}-formulas, since the given formula $\varphi=\bigwedge_{\{\,i,j\,\} \in E}x_i \vee x_j$ is a positive \classFont{2CNF}-formula.

\begin{corollary}\label{cor:chardelnp}
	Let $\mathcal{E}=\{\,\esat,\ecardmaxsat,\eminsat,\ecardminsat\,\}$.
	\begin{enumerate}
		\item For all $\prob{E}\in\mathcal{E}$ and $\varphi \in \FO(A)$ with $A\subseteq \{\,\perp,\dep{\dots},\subseteq\,\}$ $\prob{E}_\varphi^{\team}$ is in $\Del\NP$. 
		\item There are formulas $\varphi_1 \in \FO(\dep\dots), \varphi_2 \in \FO(\perp), \varphi_3 \in\FO(\subseteq)$ such that for all $\prob{E}\in\mathcal{E}$ the problems $\prob{E}_{\varphi_1}^{\team}$, $\prob{E}_{\varphi_2}^{\team}$ and $\ecardminsat_{\varphi_3}^{\team}$ are \Del\NP-complete.
	\end{enumerate}
\end{corollary}

\begin{proof}
  Statement 1.\ follows directly from Theorems~\ref{mem}, \ref{memcardmax} and \ref{memmin}.
  For statement 2., the hardness for the case of inclusion logic follows from Theorem~\ref{arbhard} together with Theorem~\ref{cardhard}, as $\cardminsat_φ^\team$ can trivially be decided in polynomial time with oracle access to $\ecardminsat_φ^\team$:
  Simply get a solution from the oracle, compute its cardinality and compare it to $k$.
  The other cases follow from Corollary~\ref{delnphard}.
\end{proof}

By Corollary~\ref{cor:chardelnp} we get a characterization of the class $\Del\NP$ as the closure of the mentioned problems under the enumeration reducibility notion.

\section{Conclusion}

In Table~\ref{tab:summary}, we summarise the complexity results we obtained in this paper. We completely classified all but one of the considered enumeration problems and obtained either polynomial-delay algorithms or completeness for $\Del\NP$. We have no final result regarding $\emaxsat_\varphi^\team$ for dependence logic and independence logic formulas. By Corollary~\ref{delnphard} this problem is $\Del\NP$-hard but we do not know if it is included in $\Del\NP.$ On the other hand the problem is included in $\Del\Sigma_2^p$, as one can construct an algorithm similar to Algorithm~\ref{alg:esat} that uses $\VerTeam_{\varphi}$ and $\ExtMaxTeam_\varphi$ as oracles (it is easy to see, that $\ExtMaxTeam_\varphi \in \Del\Sigma_2^p$). We conjecture that this problem is in fact $\Del\Sigma_2^p$-complete but we are missing the hardness proof.
\problemdef{$\ExtMaxTeam_\varphi$}{Structure $\calA, \text{ team } X, \text{ set of assignments } Y$}{$\{\,X' \mid \calA \models_{X'} \varphi, X \subsetneq X', X'\cap Y = \emptyset \text{ and } \forall X'' X' \subsetneq X'' \calA \not\models_{X''} \varphi\,\} \neq \emptyset  $}

\begin{table}\centering
\begin{tabular}{ccc}\toprule 
	& $\subseteq$ & $\dep{\dots},\, \perp$ \\\midrule
	$\esat$ & $\in \Del\P$ & $\Del\NP$-complete \\
	$\emaxsat$ & $\in \Del\P$ & $\Del\NP$-hard, $\in \Del\Sigma_2^p$ \\
	$\eminsat$ & $\in \Del\P$ & $\Del\NP$-complete \\
	$\ecardmaxsat$ & $\in \Del\P$ & $\Del\NP$-complete\\
	$\ecardminsat$ & $\Del\NP$-complete & $\Del\NP$-complete  \\\bottomrule
\end{tabular}
\captionsetup{justification=centering,margin=2cm}
\caption{Summary of obtained complexity results}\label{tab:summary}
\end{table}

There are some more open issues that immediately lead to questions for further research.
All our results are obtained for a certain fixed set of generalised dependency relations.
Our selection was motivated by those logics most frequent found in the literature.
It will be interesting to see whether other atoms or combinations of atoms lead to different (higher?) complexity.

There is a notion of \emph{strict semantics} (see, \eg, the work of Galliani~\cite{DBLP:journals/apal/Galliani12}).
Our results do not immediately transfer to strict semantics, since, for example, Lemma~\ref{vernp} is not true for independence logic with strict semantics.
It would be interesting to study the enumeration complexity of team logics in strict semantics.

Maybe even more interesting is the extension of the logical language by the so called strong (or classical) negation.
Observe that our logics only allow atomic negation.
It is known that with full classical negation, many generalised dependency atoms can be simulated (in modal logic, negation is even complete in the sense that it can simulate any FO-expressible dependency).
We consider it likely that enumeration problems for logics with classical negation will lead us out of the class $\Del\NP$ and potentially even to arbitrary levels of the hierarchy.

\clearpage
\bibliography{ref}

\newpage
\appendix

\end{document}